\documentclass[journal,twocolumn]{IEEEtran}
\usepackage[T1]{fontenc}
\usepackage{amsthm}
\usepackage{amsmath}
\usepackage[cmintegrals]{newtxmath}
\usepackage{mathtools,amssymb,bm}
\usepackage{graphicx}
\usepackage{color}
\usepackage{multicol}
\usepackage[right=0.55in,left=0.55in,top=0.45in,bottom=0.45in]{geometry}

\usepackage{lettrine}
\usepackage{mathtools}
\usepackage{footnote}
\makesavenoteenv{tabular}

\newtheorem{lem}{\textit{Lemma}}
\newtheorem{theo}{\textit{Theorem}}
\newtheorem{rem}{\textit{Remark}}

\newcommand{\dB}{\,\mathrm{dB}}
\newcommand{\tr}{\,\mathrm{tr}}
\newcommand{\as}{\,\mathrm{a.s}}
\newcommand{\dtt}{\,\mathrm{det}}

\renewcommand\IEEEkeywordsname{Index Terms}

\usepackage{siunitx}
\makeatletter
\providecommand\add@text{}
\newcommand\tagaddtext[1]{%
  \gdef\add@text{#1\gdef\add@text{}}}%
\renewcommand\tagform@[1]{%
  \maketag@@@{\llap{\add@text\quad}(\ignorespaces#1\unskip\@@italiccorr)}%
}
\makeatother

\begin{document}
\title{Jamming Detection in Massive MIMO Systems}

\author{Hossein Akhlaghpasand, S. Mohammad Razavizadeh, Emil Bj\"{o}rnson, and Tan Tai Do
\thanks{\IEEEauthorblockA{Hossein Akhlaghpasand, S. Mohammad Razavizadeh are with Iran University of Science and Technology, Iran (e-mail: h{\_}akhlaghpasand@elec.iust.ac.ir, smrazavi@iust.ac.ir). Emil Bj\"{o}rnson is with Link\"{o}ping University, Sweden (email: emil.bjornson@liu.se). Tan Tai Do is with Ericsson AB, Sweden. (e-mail: ttdo@kth.se).}}}

\maketitle

\begin{abstract}
This paper considers the physical layer security of a pilot-based massive multiple-input multiple-output (MaMIMO) system in presence of a multi-antenna jammer. To improve security of the network, we propose a new jamming detection method that makes use of a generalized likelihood ratio test over some coherence blocks. Our proposed method utilizes intentionally unused pilots in the network. The performance of the proposed detector improves by increasing the number of antennas at the base station, the number of unused pilots and also by the number of the coherence blocks that are utilized. Simulation results confirm our analyses and show that in the MaMIMO regime, perfect detection (i.e., correct detection probability is one) is achievable even with a small number of unused pilots.
\end{abstract}


\section{Introduction}\label{sec:Introduction}
Massive multiple-input multiple-output (MaMIMO) technology provides improvements of the physical layer security against passive eavesdropping \cite{r1}. To cause more harm, the eavesdropper can decide to \emph{actively} attack the legitimate system. The asymptotic secrecy rate for downlink MaMIMO systems is derived in \cite{r3} in the presence of an active multi-antenna eavesdropper. Furthermore, for detection of an active eavesdropper in a single-input multiple-output system, the authors in \cite{r4} propose a detection method based on random training. Detection methods based on cooperation between the base station (BS) and the legitimate user are also suggested in \cite{r5}. A \emph{jammer} can be viewed as an active eavesdropper, which transmits jamming signals in both the pilot and data phases. In MaMIMO systems, by jamming the pilot phase, pilot contamination is caused which significantly decreases the spectral efficiency of the legitimate users. In \cite{r6}, the authors propose a jamming detection and countermeasure scheme for MaMIMO systems based on random matrix theory. The authors in \cite{r7} propose a linear combining scheme that exploits unused pilots to improve the spectral efficiency of the MaMIMO systems in presence of a single-antenna jammer. In \cite{r8}, a jamming power allocation scheme is investigated for MaMIMO systems, where a smart jammer aims to maximize its deterioration effect on the spectral efficiency of the legitimate system.

In order to apply the known methods for combating jamming, we first need to detect the presence of the jammer. In this letter, we propose a new jamming detection method by exploiting a generalized likelihood ratio test (GLRT) in the MaMIMO systems. Our proposed detector utilizes the key properties of MaMIMO systems, i.e., many antennas at the BS and the use of uplink pilot transmission, in order to detect the jamming with high accuracy. To this end, the proposed method estimates the power of the multi-antenna jammer in some intentionally unused pilot sequences (i.e., the pilots that are not assigned to any active users, except during peak hours) over multiple coherence blocks. To evaluate performance of the proposed method, we derive closed-form expressions for the false alarm and correct detection probabilities. Our analysis shows that the jamming detection improves by increasing the number of antennas at the BS as well as the number of unused pilots. Simulation results show that in systems with a large number of BS antennas, perfect detection (i.e., correct detection probability is one) is achievable even with a small number of unused pilots.

\textit{Notations:} Matrices and vectors are denoted by boldface upper and lower case letters, respectively. $\boldsymbol{A}^*$, $\boldsymbol{A}^T$, $\boldsymbol{A}^H$, and $\tr(\boldsymbol{A})$ respectively denote conjugate, transpose, conjugate transpose and trace of a matrix $\boldsymbol{A}$. $\boldsymbol{I}_N$ is the $N \times N$ identity matrix. $\|\boldsymbol{x}\|$ denotes the Euclidean norm of a vector $\boldsymbol{x}$. Also, $\boldsymbol{0}$ and $\boldsymbol{1}$ are respectively the all-zero and all-one column vectors. $f\left( z ; s , \mathcal{H} \right)$ represents the likelihood function of a random variable $z$ under hypothesis $\mathcal{H}$ and given $s$.

\section{Problem Setup}\label{sec:Problem Setup}
In this section, we present the system model and the jammer's strategy to attack the legitimate system in the pilot phase.

\subsection{System Model}
Consider the uplink of a single-cell multi-user MIMO system with one $M_r$-antenna BS and $K$ legitimate single-antenna users. There is an $M_w$-antenna jammer in the network that aims to reduce the asymptotic spectral efficiency of the legitimate system. Each channel is modeled as Rayleigh block-fading, where it is constant within a coherence block of $T$ samples and has independent realizations between the blocks. Let $\tau$ ($\tau \leq T$) be the number of samples during a block spent on transmission of pilots for channel estimation. We assume that $\left \{\boldsymbol{\phi}_1, \ldots, \boldsymbol{\phi}_{\tau} \right \}$ is the set of $\tau$ orthonormal pilot sequences, in which $\boldsymbol{\phi}_i \in \mathbb{C}^{\tau \times 1}$ is the $i$th pilot and $\boldsymbol{\phi}_i^{T} \boldsymbol{\phi}_t^{*}=\bigg\{\begin{array}{lr}
1 , & t=i\\
0 , & t \neq i
\end{array}$. We denote by $\boldsymbol{g}_i\!\left(l\right) \in \mathbb{C}^{M_r \times 1} ,~ i=1, \ldots, K$ and $\boldsymbol{G}_w\!\left(l\right) \triangleq [\boldsymbol{g}^{\left(1\right)}_w\!\left(l\right),\ldots,\boldsymbol{g}^{\left(M_w\right)}_w\!\left(l\right)] \in \mathbb{C}^{M_r \times M_w}$ the channels from the $i$th legitimate user and the jammer to the BS in the $l$th block, respectively, which are defined by $\boldsymbol{g}_i\!\left(l\right)=\sqrt{\beta_i}\boldsymbol{h}_i\!\left(l\right)$ and $\boldsymbol{g}^{\left(j\right)}_w\!\left(l\right)=\sqrt{\beta_w}\boldsymbol{h}^{\left(j\right)}_w\!\left(l\right),~j=1,\ldots,M_w$. The scalars $\beta_i$, $\beta_w$ are the positive large-scale fading coefficients and $\boldsymbol{h}_i\!\left(l\right)$, $\boldsymbol{h}^{\left(j\right)}_w\!\left(l\right)$ represent fast fading in the $l$th block, which have independent and identically distributed (i.i.d.) $\mathcal{CN}\left(0,1\right)$ elements. Since the large-scale fading varies much slower than the fast fading, we assume that $\beta_i$, $\beta_w$ remain constant in $L$ coherence blocks which are distributed over the frequency domain and also over time.

In the pilot phase of the $l$th block, the users send their respective pilots $\boldsymbol{\phi}_i\!\left(l\right),~ i=1,\ldots,K$ to the BS. At the same time, the jammer transmits its jamming signals $\boldsymbol{\psi}^{\left(j\right)}_{w}\!\left(l\right) \in \mathbb{C}^{\tau \times 1},~j=1,\ldots,M_w$ to create pilot contamination that later will degrade the asymptotic spectral efficiency of the legitimate system. We study the single-cell scenario in this letter, which is not too different from a practical multi-cell scenario, where a pilot reuse factor of $3$, $4$, or $7$ is used \cite[Chapter 6]{r9}. Thus, the interference that comes from pilot reuse in other cells in the pilot phase is ignored (although, the interference in the data phase might be large). The received signal at the BS for a given block $l \in \left\{1,\ldots,L\right\}$ is modeled as
\begin{equation}\label{received_matrix}
\boldsymbol{Y}\!\left(l\right) = \sum_{i=1}^{K} \sqrt{\tau p} \boldsymbol{g}_{i}\!\left(l\right) \boldsymbol{\phi}^{T}_{i}\!\left(l\right) + \sqrt{\tau q} \boldsymbol{G}_w\!\left(l\right) \boldsymbol{V}_{w}\!\left(l\right) \boldsymbol{\Psi}^{T}_{w}\!\left(l\right) + \boldsymbol{N}\!\left(l\right) ,
\end{equation}
where $p$ represents the transmit powers of the users and $q$ represents the transmit power of the jammer per antenna in the pilot phase. The matrix $\boldsymbol{V}_{w}\!\left(l\right) \triangleq [\boldsymbol{v}^{{\left(1\right)}}_{w}\!\left(l\right),\ldots,\boldsymbol{v}^{{\left(M_w\right)}}_{w}\!\left(l\right)]^T \in \mathbb{C}^{M_w \times M_w}$ is the jamming precoder (which is independent of the channels to the BS) composed of the vectors $\boldsymbol{v}^{{\left(j\right)}}_{w}\!\left(l\right) \in \mathbb{C}^{M_w \times 1},~j=1,\ldots,M_w$ and $\boldsymbol{\Psi}_{w}\!\left(l\right) \triangleq [\boldsymbol{\psi}^{\left(1\right)}_{w}\!\left(l\right),\ldots,\boldsymbol{\psi}^{\left(M_w\right)}_{w}\!\left(l\right)] \in \mathbb{C}^{\tau \times M_w}$ is the jamming signal matrix. Furthermore, $\boldsymbol{N}\!\left(l\right) \in \mathbb{C}^{M_r \times \tau}$ denotes the normalized receiver noise matrix composed of the elements distributed as $\mathcal{CN}\left(0,1\right)$.

\subsection{Jammer's Strategy}
The jammer aims to maximize its negative impact on the legitimate system. Since one of the main advantages of using MaMIMO is achieving high spectral efficiency, a natural strategy that can be adopted by the jammer is to reduce the asymptotic spectral efficiency of the legitimate system. To analyze the worst-case impact of the jamming, we assume that the jammer is smart and knows the transmission protocol and the pilot set $\left \{\boldsymbol{\phi}_1, \ldots, \boldsymbol{\phi}_{\tau} \right \}$. The jammer can obtain this information by listening to the channel for some consecutive blocks. The jamming signals $\boldsymbol{\psi}^{\left(j\right)}_{w}\!\left(l\right)$ in \eqref{received_matrix} can be generally written as a combination of the orthonormal pilots as
\begin{equation}\label{Jam_pilot}
\boldsymbol{\psi}^{\left(j\right)}_{w}\!\left(l\right) = \sum_{i=1}^{\tau} \alpha^{\left(j\right)}_i\!\left(l\right) \boldsymbol{\phi}_{i} ,
\end{equation}
where $\alpha^{\left(j\right)}_i\!\left(l\right) \triangleq \boldsymbol{\psi}_{w}^{{\left(j\right)}^T}\!\left(l\right) \boldsymbol{\phi}_{i}^*$. The power constraint on the transmitted jamming signals is $\tr (\boldsymbol{\Gamma}^T_w\!\left(l\right) \boldsymbol{\Gamma}^*_w\!\left(l\right)) = M_w$, where $\boldsymbol{\Gamma}_w\!\left(l\right) \triangleq \boldsymbol{\Psi}_w\!\left(l\right) \boldsymbol{V}^T_w\!\left(l\right)$. By defining $M_w \times 1$ vectors of coefficients $\boldsymbol{\alpha}_i\!\left(l\right) \triangleq [\alpha^{\left(1\right)}_i\!\left(l\right),\ldots,\alpha^{\left(M_w\right)}_i\!\left(l\right)]^T,~i=1,\ldots,\tau$, the $j$th column of $\boldsymbol{\Gamma}_w\!\left(l\right)$ is derived as
\begin{equation}\label{Prod_jam_precod_pilot}
\boldsymbol{\gamma}^{\left(j\right)}_{w}\!\left(l\right) = \sum_{i=1}^{\tau} \left(\boldsymbol{v}^{\left(j\right)^T}_{w}\!\left(l\right)\boldsymbol{\alpha}_{i}\!\left(l\right)\right) \boldsymbol{\phi}_{i} .
\end{equation}
Utilizing \eqref{Prod_jam_precod_pilot}, the power constraint on the transmitted jamming signals is given by
\begin{equation}\label{power_cons_jam}
\sum_{i=1}^{\tau} \sum_{j=1}^{M_w} \left|\boldsymbol{v}^{\left(j\right)^T}_{w}\!\left(l\right)\boldsymbol{\alpha}_{i}\!\left(l\right)\right|^2 = M_w
\end{equation}
after some algebraic simplifications.
\begin{figure*}[!t]
\setcounter{equation}{14}
\begin{equation}\label{log_likelihood_all}
\ln \left(f\left(\left\{\boldsymbol{Y\!}_w\!\left(l\right)\right\}_{l=1}^{L} ; \tilde{q} , \mathcal{H}_{1} \right)\right) = -\frac{M_rL\left(\tau-K\right)}{2} \ln \left(2\pi \right) - \frac{M_rL}{2} \ln \left(1+\left(\tau-K\right)\tilde{q} \right) 
- \frac{1}{2} \sum_{l=1}^{L} \sum_{m=1}^{M_r} \left(\|\tilde{\boldsymbol{y}}_m\!\left(l\right)\|^2-\frac{\tilde{q}|\tilde{\boldsymbol{y}}_m^H\!\left(l\right) \boldsymbol{1}|^2}{1+\left(\tau-K\right)\tilde{q}} \right)
\end{equation}
\hrulefill
\end{figure*}
\begin{lem}\label{lem_1}
(\textbf{Asymptotic spectral efficiency}) In a multi-user MIMO system by utilizing the minimum mean-squared error estimate in the pilot phase and the maximum ratio combiner in the data phase as well as for any jamming precoder in the data phase that is independent of the channels to the BS and follows the power constraint on the transmitted jamming signals, the asymptotic spectral efficiency is defined by $\mathbb{E}\left\{C^{(\text{asy})}\!\left(l\right)\right\}$ where\footnote{We use the notation $\mathcal{J}_1\!\left(M_r\right) \asymp \mathcal{J}_2\!\left(M_r\right)$ for two arbitrary functions $\mathcal{J}_1\left(M_r\right)$ and $\mathcal{J}_2\!\left(M_r\right)$, if $\mathcal{J}_1\!\left(M_r\right) - \mathcal{J}_2\!\left(M_r\right) \xrightarrow[M_r \rightarrow \infty]{\as} 0$.}
\setcounter{equation}{4}
\begin{equation}\label{ASE}
C^{(\text{asy})}\!\left(l\right) \asymp \left(1-\frac{\tau}{T}\right) \sum_{i=1}^{K} \log_2 \left(1+\frac{\frac{p}{q} \frac{\rho}{\varrho} \left(\frac{\beta_i}{\beta_w} \right)^2}{M_w \sum_{j=1}^{M_w} \left|\boldsymbol{v}^{\left(j\right)^T}_{w}\!\left(l\right)\boldsymbol{\alpha}_{i}\!\left(l\right)\right|^2} \right) .
\end{equation}
The scalar $\rho$ represents the average transmit powers of the users and $\varrho$ represents the average transmit power of the jammer per antenna in the data phase. The necessary condition for \eqref{ASE} is $\sum_{j=1}^{M_w} |\boldsymbol{v}^{\left(j\right)^T}_{w}\!\left(l\right)\boldsymbol{\alpha}_{i}\!\left(l\right)|^2 > 0$ for all $i \in \left\{1,\ldots,K\right\}$, otherwise $C^{(\text{asy})}\!\left(l\right) \rightarrow \infty$.
\end{lem}
\begin{proof}
Lemma \ref{lem_1} is derived analogously to the analysis presented in \cite[Section II]{r8}, which is omitted here due to limit of space.
\end{proof}

The number of pilots is fixed in practice, while the number of users vary and is usually below the peak value \cite{r10}. Hence, there are generally $\tau-K \geq 1$ unused pilots in the network which can be exploited to improve the security of the network. To this end, we assume that the BS uses a pseudo-random pilot hopping technique to assign the pilots randomly to the users \cite{r10}. This technique prevents the jammer from estimating which pilots are assigned to the users.
\begin{rem}
(\textbf{Jammer's best strategy}) In general, the jammer's strategy depends on its objective and available information. As we said before, here the jammer aims to prevent the asymptotic spectral efficiency in \eqref{ASE} from unlimitedly increasing. If at least one of $\sum_{j=1}^{M_w} |\boldsymbol{v}^{\left(j\right)^T}_{w}\!\left(l\right)\boldsymbol{\alpha}_{i}\!\left(l\right)|^2 = 0 , ~ i=1,\ldots,K$, then $C^{(\text{asy})}\!\left(l\right) \rightarrow \infty$. This means that the jammer will not successfully limit the asymptotic spectral efficiency of the MaMIMO systems, if it does not attack all users. Furthermore, the jammer does not know the current assigned pilot to each user. Hence, the best strategy that it can adopt, is to divide its power equally over all pilots, i.e.,
\begin{equation}\label{Uniformly_power_divide_V1}
\sum_{j=1}^{M_w} \left|\boldsymbol{v}^{\left(j\right)^T}_{w}\!\left(l\right)\boldsymbol{\alpha}_{1}\!\left(l\right)\right|^2 = \ldots = \sum_{j=1}^{M_w} \left|\boldsymbol{v}^{\left(j\right)^T}_{w}\!\left(l\right)\boldsymbol{\alpha}_{\tau}\!\left(l\right)\right|^2 .
\end{equation}
From this strategy and the power constraint in \eqref{power_cons_jam}, we have
\begin{equation}\label{Uniformly_power_divide}
\sum_{j=1}^{M_w} \left|\boldsymbol{v}^{\left(j\right)^T}_{w}\!\left(l\right)\boldsymbol{\alpha}_{i}\!\left(l\right)\right|^2 = \frac{M_w}{\tau} .
\end{equation}
\end{rem}

\subsection{Utilization of Unused Pilots}
By substituting the definition of the matrix $\boldsymbol{\Gamma}_w\!\left(l\right)$ into \eqref{received_matrix} and using \eqref{Prod_jam_precod_pilot}, the received signal at the BS can be rewritten as
\begin{multline}\label{received_mat_simpler}
\boldsymbol{Y}\!\left(l\right) = \sum_{i=1}^{K} \left(\sqrt{\tau p} \boldsymbol{g}_{i}\!\left(l\right) + \sum_{j=1}^{M_w} \sqrt{\tau q} \left(\boldsymbol{v}^{\left(j\right)^T}_{w}\!\left(l\right)\boldsymbol{\alpha}_{i}\!\left(l\right)\right) \boldsymbol{g}^{\left(j\right)}_{w}\!\left(l\right) \right) \boldsymbol{\phi}^T_{i} \\
+ \sum_{i=K+1}^{\tau} \sum_{j=1}^{M_w} \sqrt{\tau q} \left(\boldsymbol{v}^{\left(j\right)^T}_{w}\!\left(l\right)\boldsymbol{\alpha}_{i}\!\left(l\right)\right) \boldsymbol{g}^{\left(j\right)}_{w}\!\left(l\right) \boldsymbol{\phi}^T_{i} + \boldsymbol{N}\!\left(l\right) .
\end{multline}
For ease of exposition and without loss of generality, it is assumed in \eqref{received_mat_simpler} that the first $K$ pilots of the pilot set are assigned to the legitimate users and the remaining pilots are unused, i.e., $\boldsymbol{\phi}_1\!\left(l\right)=\boldsymbol{\phi}_1,\ldots, \boldsymbol{\phi}_K\!\left(l\right)=\boldsymbol{\phi}_{K}$. Since the signal received along each of the unused pilots includes only the jamming signals and noise, we exploit it to detect the jamming attack. By projecting $\boldsymbol{Y}\!\left(l\right)$ on each of the unused pilots, we have
\begin{multline}\label{Project_rec_sig}
\boldsymbol{y}_{i}\!\left(l\right) = \boldsymbol{Y}\!\left(l\right) \boldsymbol{\phi}_i^* = \sum_{j=1}^{M_w} \sqrt{\tau q} \left(\boldsymbol{v}^{\left(j\right)^T}_{w}\!\left(l\right)\boldsymbol{\alpha}_{i}\!\left(l\right)\right) \boldsymbol{g}^{\left(j\right)}_{w}\!\left(l\right) + \boldsymbol{n}_{i}\!\left(l\right) , \\ i=K+1,\ldots,\tau,
\end{multline}
where $\boldsymbol{n}_{i}\!\left(l\right)=\boldsymbol{N}\!\left(l\right) \boldsymbol{\phi}_i^*$ is the projection of the matrix $\boldsymbol{N}\!\left(l\right)$ on the $i$th pilot and $\boldsymbol{n}_{i}\!\left(l\right) \sim \mathcal{CN} \left( \boldsymbol{0} , \boldsymbol{I}_M \right)$. In the next section, we investigate problem of the jamming detection using the observed signals on the unused pilots in \eqref{Project_rec_sig}.

\section{Jamming Detection}\label{sec:Jam Det}
In this section, detection of the jamming attack is considered based on the likelihood functions of the signals observed on the unused pilots' directions. The hypothesis test is defined as follows
\begin{align}\label{hypo_test}
&\mathcal{H}_{0}:
\boldsymbol{y}_{i}\!\left(l\right) = \boldsymbol{n}_{i}\!\left(l\right) , \nonumber\\
&\mathcal{H}_{1}:
\boldsymbol{y}_{i}\!\left(l\right) = \sum_{j=1}^{M_w} \sqrt{\tau q} \left(\boldsymbol{v}^{\left(j\right)^T}_{w}\!\left(l\right)\boldsymbol{\alpha}_{i}\!\left(l\right)\right) \boldsymbol{g}^{\left(j\right)}_{w}\!\left(l\right) + \boldsymbol{n}_{i}\!\left(l\right) .
\end{align}
The hypotheses $\mathcal{H}_{0}$ and $\mathcal{H}_{1}$ respectively denote the absence and presence of the jammer. The vectors $\boldsymbol{g}^{\left(j\right)}_{w}\!\left(l\right)$ and $\boldsymbol{n}_i\!\left(l\right)$ have Gaussian distributions. Let $\boldsymbol{Y\!}_w\!\left(l\right) \triangleq \left[\boldsymbol{y}_{K+1}\!\left(l\right),\ldots,\boldsymbol{y}_{\tau}\!\left(l\right) \right]$ be an $M_r \times \left(\tau-K\right)$ matrix. Under $\mathcal{H}_1$, the same terms $\boldsymbol{g}^{\left(j\right)}_w\!\left(l\right),~j=1,\ldots,M_w$ are repeated in all columns of $\boldsymbol{Y\!}_w\!\left(l\right)$. Therefore, there is a correlation between the elements of each row of the matrix $\boldsymbol{Y\!}_w\!\left(l\right)$. If the $m$th column of $\boldsymbol{Y\!}_w^T\!\left(l\right)$ is denoted by $\tilde{\boldsymbol{y}}_{m}\!\left(l\right) \in \mathbb{C}^{\left(\tau-K\right) \times 1}$, its covariance matrix is calculated as
\begin{equation}
\boldsymbol{C}\left(\tilde{q}\right) = \mathbb{E}\left\{\tilde{\boldsymbol{y}}_{m}\!\left(l\right) \tilde{\boldsymbol{y}}_{m}^H\!\left(l\right)\right\} = \boldsymbol{I}_{\left(\tau-K\right)} + \tilde{q}~ \boldsymbol{1}\boldsymbol{1}^T , \nonumber
\end{equation}
where $\tilde{q} \triangleq q M_w \beta_w$. The likelihood function of $\tilde{\boldsymbol{y}}_{m}\!\left(l\right)$ is obtained as
\begin{equation}\label{likelihood_row}
f\left(\tilde{\boldsymbol{y}}_{m}\!\left(l\right) ; \tilde{q} , \mathcal{H}_1 \right) = \frac{\exp \left(-\frac{1}{2}\tilde{\boldsymbol{y}}^H_{m}\!\left(l\right) \boldsymbol{C}^{-1}\left(\tilde{q}\right) \tilde{\boldsymbol{y}}_{m}\!\left(l\right) \right)}{\sqrt{\left(2\pi \right)^{\tau-K}\dtt \left(\boldsymbol{C}\left(\tilde{q}\right)\right)}} .
\end{equation}
Since the rows of $\boldsymbol{Y\!}_w\!\left(l\right)$ are independent from each other and also from the elements of $\boldsymbol{Y\!}_w\!\left(l'\right)$ for $l \neq l'$, we have
\begin{equation}\label{likelihood_all_H1}
f\left(\left\{\boldsymbol{Y\!}_w\!\left(l\right)\right\}_{l=1}^{L} ; \tilde{q} , \mathcal{H}_1 \right) = \prod_{l=1}^{L} \prod_{m=1}^{M_r}\frac{\exp \left(-\frac{1}{2}\tilde{\boldsymbol{y}}_{m}^H\!\left(l\right) \boldsymbol{C}^{-1}\left(\tilde{q}\right) \tilde{\boldsymbol{y}}_{m}\!\left(l\right) \right)}{\sqrt{\left(2\pi \right)^{\tau-K}\dtt\left(\boldsymbol{C}\left(\tilde{q}\right)\right)}} .
\end{equation}
Considering \eqref{hypo_test}, the likelihood function under $\mathcal{H}_0$ is obtained by setting $\tilde{q} = 0$ in \eqref{likelihood_all_H1}. Hence, we have 
\begin{equation}\label{likelihood_all_H0}
f\left(\left\{\boldsymbol{Y\!}_w\!\left(l\right)\right\}_{l=1}^{L} ; \mathcal{H}_0 \right) = \prod_{l=1}^{L} \prod_{m=1}^{M_r}\frac{\exp \left(-\frac{1}{2}\|\tilde{\boldsymbol{y}}_{m}\!\left(l\right)\|^2\right)}{\sqrt{\left(2\pi \right)^{\tau-K}}} .
\end{equation}
Since in \eqref{likelihood_all_H1}, $\tilde{q}$ is an unknown parameter, we exploit the GLRT \cite{r11} to decide which hypothesis is true. This enables us to detect the jamming. The detector derived from the GLRT is given by
\begin{equation}\label{glrt_test}
\mathcal{L}\left(\left\{\boldsymbol{Y\!}_w\!\left(l\right)\right\}_{l=1}^{L}\right) = \frac{f\left(\left\{\boldsymbol{Y\!}_w\!\left(l\right)\right\}_{l=1}^{L} ; \tilde{q} , \mathcal{H}_1 \right)}{f\left(\left\{\boldsymbol{Y\!}_w\!\left(l\right)\right\}_{l=1}^{L} ; \mathcal{H}_0 \right)} \mathop{\gtrless}_{\mathcal{H}_0}^{\mathcal{H}_1} \mu ,
\end{equation}
where $\mu$ is a threshold parameter for detecting the presence of the jammer. In order to implement the test in \eqref{glrt_test}, the BS estimates $\tilde{q}$ by utilizing maximum likelihood (ML) estimation \cite{r12} under $\mathcal{H}_1$. An equivalent expression for \eqref{likelihood_all_H1} is obtained as \eqref{log_likelihood_all}, shown at the top of this page. The ML estimation of $\tilde{q}$ is performed by maximizing the log-likelihood function. Hence, it is obtained by differentiating \eqref{log_likelihood_all} with respect to $\tilde{q}$ and solving the following equation
\setcounter{equation}{15}
\begin{equation}\label{zero_eq_mle}
\frac{M_rL\left(\tau-K\right)}{1+\left(\tau-K\right)\tilde{q}} - \sum_{l=1}^{L} \sum_{m=1}^{M_r} \frac{|\tilde{\boldsymbol{y}}_{m}^H\!\left(l\right) \boldsymbol{1}|^2}{\left(1+\left(\tau-K\right)\tilde{q} \right)^2}=0.
\end{equation}
Therefore, the ML estimate of $\tilde{q}$, which is denoted by $\hat{\tilde{q}}$, is
\begin{equation}\label{mle_jam_power}
\hat{\tilde{q}} = \frac{\sum_{l=1}^{L} \sum_{m=1}^{M_r}|\tilde{\boldsymbol{y}}_{m}^H\!\left(l\right) \boldsymbol{1}|^2}{M_rL\left(\tau-K\right)^2} - \frac{1}{\tau-K}.
\end{equation}
Replacing $\tilde{q}$ in \eqref{glrt_test} by its ML estimate from \eqref{mle_jam_power}, the test is calculated in logarithmic form as
\begin{align}\label{glrt_calc}
\ln\left(\mathcal{L}\left(\left\{\boldsymbol{Y\!}_w\!\left(l\right)\right\}_{l=1}^{L} \right)\right) &= -\frac{M_rL}{2} \ln \left(1+\left(\tau-K\right)\hat{\tilde{q}} \right) \nonumber\\
& \quad + \frac{\hat{\tilde{q}}\sum_{l=1}^{L} \sum_{m=1}^{M_r}|\tilde{\boldsymbol{y}}_{m}^H\!\left(l\right) \boldsymbol{1}|^2}{2\left(1 + \left(\tau-K\right) \hat{\tilde{q}} \right)} \mathop{\gtrless}_{\mathcal{H}_0}^{\mathcal{H}_1} \ln \left(\mu \right) .
\end{align}
Some manipulations of the above inequality lead to
\begin{equation}\label{glrt_calc_eq2}
- \ln \left(1+\left(\tau-K\right)\hat{\tilde{q}} \right) + \left(\tau-K\right) \hat{\tilde{q}} \mathop{\gtrless}_{\mathcal{H}_0}^{\mathcal{H}_1} \frac{2}{M_rL}\ln \left(\mu \right) .
\end{equation}
Since $\tilde{q}$ must be nonnegative, the negative values of $\hat{\tilde{q}}$ in \eqref{mle_jam_power} are ignored and replaced by $\hat{\tilde{q}}=0$. By defining a monotonically increasing function $\mathcal{J}\!\left(x\right) \triangleq x - \ln \left(1+x\right)$ for $x \geq 0$ and using \eqref{glrt_calc_eq2}, we have
\begin{equation}\label{glrt_calc_eq3}
\mathcal{J}\!\left(\left(\tau-K\right)\hat{\tilde{q}}\right) \mathop{\gtrless}_{\mathcal{H}_0}^{\mathcal{H}_1} \frac{2}{M_rL}\ln \left(\mu \right) ,
\end{equation}
or equivalently
\begin{equation}\label{glrt_calc_eq3}
\hat{\tilde{q}}~ \mathop{\gtrless}_{\mathcal{H}_0}^{\mathcal{H}_1} \frac{1}{\left(\tau-K\right)} \mathcal{J}^{-1}\!\left(\frac{2}{M_rL}\ln \left(\mu \right)\right) = \mu' .
\end{equation}
Therefore, we can define a positive threshold $\mu'$ and compare $\hat{\tilde{q}} \mathop{\gtrless}_{\mathcal{H}_0}^{\mathcal{H}_1} \mu'$ for detecting the jamming attacks.

\subsection{False Alarm Probability}
When the jammer is not present in the network but the BS detects its presence by mistake, a ``false alarm'' will occur. The false alarm probability is defined as $P_{FA} \triangleq \Pr \left( \hat{\tilde{q}} > \mu' ; \mathcal{H}_0 \right)$.
\begin{theo}\label{theo_1}
In MaMIMO systems, the false alarm probability of the proposed jamming detector is
\begin{equation}\label{false_prob}
P_{FA} = 1 - \frac{\gamma \left(M_rL, \frac{M_rL\left(\tau-K\right) \mu' + M_rL}{2} \right)}{\Gamma \left(M_rL \right)} ,
\end{equation}
where $\Gamma\left(\cdot \right)$ and $\gamma\left(\cdot , \cdot \right)$ are the gamma and lower incomplete gamma functions, respectively. The false alarm probability in \eqref{false_prob} asymptotically behaves as
\begin{equation}\label{false_asymp}
P_{FA} \asymp Q\left(\sqrt{M_rL} \mu' \left(\tau-K\right) \right) ,
\end{equation}
when $M_r \rightarrow \infty$. The notation $Q\left(\cdot \right)$ represents the Q-function.
\end{theo}
\begin{proof}
Under $\mathcal{H}_0$, $\tilde{\boldsymbol{y}}_{m}^H\!\left(l\right) \boldsymbol{1}$ is a zero-mean Gaussian random variable with variance $\tau-K$. Therefore, $r \triangleq \frac{1}{\tau-K} \sum_{l=1}^{L} \sum_{m=1}^{M_r}|\tilde{\boldsymbol{y}}_{m}^H\!\left(l\right) \boldsymbol{1}|^2$ is distributed as a chi-square random variable with $2M_rL$ degrees of freedom. Hence, we can obtain the false alarm probability
\begin{align}\label{false_proof}
P_{FA} &\triangleq \Pr \left( \hat{\tilde{q}} > \mu' ; \mathcal{H}_0 \right)\nonumber \\
&= \Pr \left( r > M_rL\left(\tau-K\right) \mu' + M_rL \right) \nonumber\\
&= 1 - \frac{\gamma \left(M_rL , \frac{M_rL\left(\tau-K\right) \mu' + M_rL}{2} \right)}{\Gamma \left(M_rL \right)}.
\end{align}
The variable $r$ converges in distribution to a Gaussian random variable as the number of degrees of freedom tends to infinity. Therefore, $\hat{\tilde{q}}$ converges in probability to $\mathcal{N} \left(0 , \frac{1}{M_rL\left(\tau-K\right)^2} \right)$ as $M_r \rightarrow \infty$. According to this distribution, we have $P_{FA} \asymp Q(\sqrt{M_rL} \mu' (\tau-K))$.
\end{proof}
The false alarm probability is a criterion to choose an appropriate value for $\mu'$. For the chosen $\mu'$, we obtain the correct detection probability in next subsection.

\subsection{Correct Detection Probability}
If the BS successfully detects the presence of the jammer, a ``correct detection'' will occur and its probability is defined as $P_C \triangleq \Pr \left( \hat{\tilde{q}} >\mu' ; \mathcal{H}_1 \right)$.
\begin{theo}\label{theo_2}
In MaMIMO systems, the correct detection probability of the proposed jamming detector is
\begin{equation}\label{correct_prob}
P_{C} = 1 - \frac{\gamma \left(M_rL , \frac{M_rL\left(\tau-K\right) \mu' + M_rL}{2\left(1+\left(\tau-K\right)\tilde{q}\right)} \right)}{\Gamma \left(M_rL \right)} ,
\end{equation}
which asymptotically behaves as
\begin{equation}\label{correct_asymp}
P_C \asymp 1-Q\left(\sqrt{M_rL} \frac{\left(\tilde{q} - \mu' \right) \left(\tau-K\right)}{1+\left(\tau-K\right) \tilde{q}} \right)
\end{equation}
when $M_r \rightarrow \infty$.
\end{theo}
\begin{proof}
The proof is similar to Theorem \ref{theo_1}, except that under $\mathcal{H}_1$, $\tilde{\boldsymbol{y}}_{m}^H\!\left(l\right) \boldsymbol{1}$ has the variance $\left(\tau-K\right)\left(1+ \left(\tau-K\right) \tilde{q} \right)$.
\end{proof}
\begin{rem}
Since $\tilde{q} \geq 0$ and provided that $\tilde{q} > \mu'$, from \eqref{correct_asymp}, we see that the correct detection probability of the proposed jamming detector increases by increasing the number of BS antennas. In particular, $P_C \rightarrow 1$ and $P_{FA} \rightarrow 0$ as $M_r \rightarrow \infty$.
\end{rem}

\section{Numerical Results}\label{sec:Simulation}
In this section, we numerically evaluate the behavior of the proposed detector. It is assumed that $M_w = 4$, $\beta_w = 1$, $\tau=10$ and also the number of independent runs for Monte Carlo simulations is $10^5$. We select the matrices $\boldsymbol{V}_w\!\left(l\right)$ and $\boldsymbol{\Gamma}_w\!\left(l\right)$ so that the constraint in \eqref{Uniformly_power_divide} is satisfied.

\begin{figure}
\centering
\includegraphics[width=0.42\textwidth]
{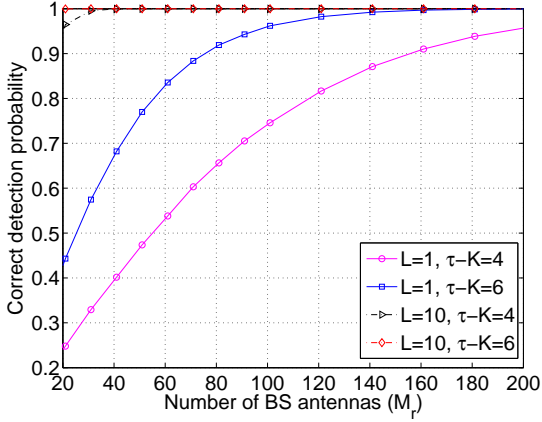}
\caption{Correct detection probability versus the number of BS antennas for different number of unused pilots and coherence blocks (for $P_{FA}=0.01$, and $q=-17 \dB$).}
\label{fig1}
\end{figure}

Fig. \ref{fig1} shows how the correct detection probability depends on the number of antennas, unused pilots, and coherence blocks that the detection is performed over, simulated under $P_{FA} = 0.01$ and $q=-17\dB$.\footnote{Since the noise variance is normalized to one, $q$ can be interpreted as the SNR.} The curves are plotted for two different number of users $K \in \{6,\,4\}$ (leading to the number of unused pilots $\tau-K \in \{4,\,6\}$) and two different number of the coherence blocks $L \in \{1,\,10\}$. It can be seen that by increasing the number of BS antennas as well as the number of the coherence blocks, the correct detection probability improves. We see that nearly perfect detection can be achieved, for example, by utilizing six unused pilots and $M_r=200$ antennas for $L=1$. Similarly, for $M_r=20$ and $L=10$, the number of required unused pilots for achieving perfect detection is set to six. However, for $M_r=50$ and $L=10$, perfect detection is achieved with only four unused pilots.

In Fig. \ref{fig2}, the correct detection probability of the proposed detector is depicted versus the false alarm probability for different values of the jammer's transmission power. It is assumed $M_r=100$, $L=10$ and there are two unused pilot sequences in the network ($K=8$). Moreover, when the value of $\mu'$ is lower than zero, nullifying the negative values of $\hat{\tilde{q}}$ leads to both the false alarm and correct detection probabilities become one. According to \eqref{false_asymp}, $\mu' < 0$ for $P_{FA} > 0.5$. Hence, the correct detection probability is plotted only for $P_{FA} \in \left[0,0.5\right]$. It can be seen in Fig. \ref{fig2} that the correct detection probability improves by increasing the jammer's transmit power. In addition, all the curves are above the line $P_C=P_{FA}$, even for low power at the jammer. Hence, the proposed detector behaves very well in MaMIMO systems. Furthermore, this simulation confirms our
asymptotic analyses as $M_r \rightarrow \infty$ which were presented in the previous section.

\begin{figure}
\centering
\includegraphics[width=0.42\textwidth]
{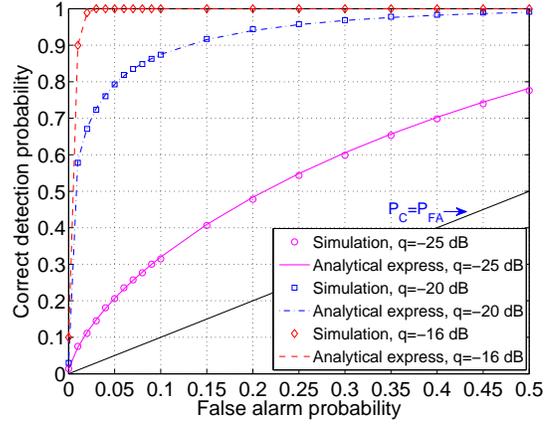}
\caption{Correct detection probability versus false alarm probability for different powers at the jammer (for $M_r=100$, $K=8$, and $L=10$).}
\label{fig2}
\end{figure}

\section{Conclusion}\label{sec:conclusion}
In this paper, a new jamming detector was proposed based on a GLRT to improve the security of MaMIMO networks. The proposed detector uses the unused pilots to estimate the received power of the multi-antenna jammer over multiple coherence blocks over the frequency domain, and potentially also time domain. Our analysis shows that the proposed jamming detector works particularly well in MaMIMO systems and it can be further enhanced by increasing the number of BS antennas, unused pilots, and also the coherence blocks that the detection is performed over. While prior works have proposed methods to mitigate jamming in MaMIMO, in practice, our detector can be used as a first step to determine if there exists any jamming that the legitimate system needs to mitigate.

\bibliographystyle{IEEE}

\begin{thebibliography}{1}

\bibitem{r1}
D. Kapetanovi\'{c}, G. Zheng, and F. Rusek,
\newblock ``Physical layer security for massive MIMO: An overview on passive evesdropping and active attacks,''
\newblock {\em IEEE Commun. Mag.}, vol. 53, no. 6, pp. 21-27, June 2015.


\bibitem{r3}
Y. Wu, R. Schober, D. W. K. Ng, C. Xiao, and G. Caire, \newblock ``Secure massive MIMO transmission with an active eavesdropper,'' \newblock{\em IEEE Trans. Inf. Theory}, vol. 62, no. 7, pp. 3880-3900, Jul. 2016.

\bibitem{r4}
D. Kapetanovi\'{c}, G. Zheng, K. K. Wong, and B. Ottersten, \newblock ``Detection of pilot contamination attack using random training and massive MIMO,'' in \newblock{\em Proc. IEEE Int. Symp. Pers. Indoor Mobile Radio Commun. (PIMRC)}, 2013, pp. 13-18.

\bibitem{r5}
D. Kapetanovi\'{c}, A. Al-Nahari, A. Stojanovic, and F. Rusek, \newblock ``Detection of active eavesdroppers in massive MIMO,'' in \newblock{\em Proc. IEEE Int. Symp. Pers. Indoor Mobile Radio Commun. (PIMRC)}, 2014, pp. 585-589.

\bibitem{r6}
J. Vinogradova, E. Bj\"{o}rnson, and E. G. Larsson, \newblock ``Detection and mitigation of jamming attacks in massive MIMO systems using random matrix theory,'' in \newblock{\em Proc. IEEE Int. Workshop on Sig.  Proc. Adv. in Wireless Commun. (SPAWC)}, 2016, pp. 1-5.

\bibitem{r7}
T. T. Do, E. Bj\"{o}rnson, E. G. Larsson, and S. M. Razavizadeh, \newblock ``Jamming-resistant receivers for massive MIMO uplink,'' in \newblock{\em IEEE Trans. Inf. Forensics Security}, vol. PP, no. 99, pp. 1-1. \\
doi: 10.1109/TIFS.2017.2746007

\bibitem{r8}
H. Pirzadeh, S. M. Razavizadeh, and E. Bj\"{o}rnson, \newblock ``Subverting massive MIMO by smart jamming,'' \newblock {\em IEEE Wireless Commun. Lett.}, vol. 5, no. 1, pp. 20-23, Feb. 2016.

\bibitem{r9}
E. G. Larsson, H. Q. Ngo, T. L. Marzetta, and Y. Hong, \newblock {\em Fundamentals of Massive MIMO}, Cambridge, NJ: Cambridge Univ. Press, 2016.

\bibitem{r10}
Y. O. Basciftci, C. E. Koksal, and A. Ashikhmin, \newblock ``Securing massive MIMO at the physical layer,'' in \newblock{\em IEEE Conf. on Commun. and Net. Sec.
(CNS)}, 2015, pp. 272-280.

\bibitem{r11}
S. M. Kay, \emph{Fundamentals of Statistical Signal Processing: Detection Theory.} Englewood Cliffs, NJ:~Prentice-Hall, 1998.

\bibitem{r12}
S. M. Kay, \newblock {\em Fundamentals of Statistical Signal Processing: Estimation Theory.} Englewood Cliffs, NJ: Prentice-Hall, 1993.

\end{thebibliography}

\end{document}